\documentclass[11pt]{article}

\usepackage{palatino}
\usepackage{mathpazo}

\usepackage{braket}
\usepackage{amsfonts}
\usepackage{amssymb}
\usepackage{amsmath}
\usepackage{latexsym}
\usepackage{amsthm}
\usepackage[usenames]{color}
\usepackage{hyperref}
\usepackage{gastex}

\hypersetup{pdfpagemode=UseNone}

\newtheorem{theorem}{Theorem}[section]
\newtheorem{lemma}[theorem]{Lemma}

\theoremstyle{definition}

\newtheorem{remark}[theorem]{Remark}

\newenvironment{mylist}[1]{\begin{list}{}{
	\setlength{\leftmargin}{#1}
	\setlength{\rightmargin}{0mm}
	\setlength{\labelsep}{2mm}
	\setlength{\labelwidth}{8mm}
	\setlength{\itemsep}{0mm}}}
	{\end{list}}

\newcommand{\tinyspace}{\mspace{1mu}}

\newcommand{\op}[1]{\operatorname{#1}}
\newcommand{\abs}[1]{\left\lvert\tinyspace #1 \tinyspace\right\rvert}
\newcommand{\ceil}[1]{\left\lceil #1 \right\rceil}

\newcommand{\norm}[1]{\left\lVert\tinyspace #1 \tinyspace\right\rVert}
\newcommand{\snorm}[1]{\lVert\tinyspace#1\tinyspace\rVert}
\newcommand{\tr}{\operatorname{Tr}}

\newcommand{\ip}[2]{\left\langle #1 , #2\right\rangle}
\newcommand{\fid}{\operatorname{F}}
\newcommand{\setft}[1]{\mathrm{#1}}
\newcommand{\lin}[1]{\setft{L}\left(#1\right)}
\newcommand{\density}[1]{\setft{D}\left(#1\right)}
\newcommand{\unitary}[1]{\setft{U}\left(#1\right)}

\newcommand{\pos}[1]{\setft{Pos}\left(#1\right)}
\newcommand{\class}[1]{\textup{#1}}

\def\complex{\mathbb{C}}
\def\real{\mathbb{R}}

\def\I{\mathbb{1}}
\def\({\left(}
\def\){\right)}
\def\X{\mathcal{X}}
\def\Y{\mathcal{Y}}
\def\Z{\mathcal{Z}}
\def\W{\mathcal{W}}
\def\yes{\text{yes}}
\def\no{\text{no}}

\begin{document}

\title{\bf Two-message quantum interactive proofs are in PSPACE}

\author{%
  Rahul Jain\thanks{Department of Computer Science and Centre for
    Quantum Technologies, National University of Singapore.}
  \and
  Sarvagya Upadhyay\thanks{%
    School of Computer Science and Institute for Quantum Computing,
    University of Waterloo.}
  \and
  John Watrous${}^{\dagger}$
}

\date{May 8, 2009}

\maketitle

\vspace{-5mm}

\begin{abstract}
  We prove that $\class{QIP}(2)$, the class of problems having
  two-message quantum interactive proof systems, is a subset of
  $\class{PSPACE}$.
  This relationship is obtained by means of an efficient parallel
  algorithm, based on the multiplicative weights update method,
  for approximately solving a certain class of semidefinite programs.
\end{abstract}

\section{Introduction}

Since their introduction roughly 25 years ago
\cite{Babai85,GoldwasserMR85}, interactive
proof systems have become a fundamental notion in the theory of
computational complexity.
The expressive power of one of the most basic variant of the
interactive proof system model, wherein a polynomial-time
probabilistic verifier interacts with a computationally unbounded
prover for a polynomial number of rounds, is characterized
\cite{LundFKN92,Shamir92} by the well-known relationship
\[
\class{IP} = \class{PSPACE}.
\]
Many variants of interactive proof systems have been studied,
including public-coin interactive proof systems (or Arthur--Merlin
games) \cite{Babai85,BabaiM88,GoldwasserS89}, zero-knowledge
interactive proofs \cite{GoldwasserMR89,GoldreichMW91} and multi-prover
interactive proofs \cite{Ben-OrGKW88}.

This paper is concerned primarily with quantum interactive proof
systems, which are defined in a similar way to ordinary interactive
proof systems except that the prover and verifier may perform quantum
computations.
Like their classical analogues, several variants of quantum
interactive proof systems have been studied, including ordinary
quantum interactive proofs \cite{Watrous03-pspace,KitaevW00},
public-coin quantum interactive proofs \cite{MarriottW05},
zero-knowledge quantum interactive proofs
\cite{Watrous02,Watrous06,Kobayashi08,HallgrenKSZ08}, and multi-prover
quantum interactive proofs \cite{KobayashiM03,KempeKMV08}.
The complexity class $\class{QIP}$ of problems having quantum
interactive proof systems is known \cite{KitaevW00} to satisfy
\[
\class{PSPACE} \subseteq \class{QIP} \subseteq \class{EXP}.
\]
The containment $\class{QIP}\subseteq\class{EXP}$ follows from the
existence of polynomial-time algorithms for approximately solving
semidefinite programs \cite{GrotschelLS93}.
(Somewhat simpler proofs of the containment
$\class{QIP}\subseteq\class{EXP}$ follow from the results of
\cite{Ben-AroyaT09,Watrous09-cb-sdp}, but these proofs still require
efficient algorithms for solving convex/semidefinite programs.)

Quantum interactive proof systems have an interesting property that
classical interactive proof systems are conjectured not to hold, which
is that they can be parallelized to a constant number of rounds of
interaction \cite{KitaevW00}.
More precisely, it holds that $\class{QIP}(3) = \class{QIP}$, where
in general $\class{QIP}(m)$ denotes the class of problems having
quantum interactive proof systems in which $m$ messages are exchanged
between the prover and verifier (with the prover always sending the
last message).
This leaves four basic classes that are defined naturally by quantum
interactive proof systems: $\class{QIP}(0) = \class{BQP}$,
$\class{QIP}(1) = \class{QMA}$, $\class{QIP}(2)$, and $\class{QIP}(3)
= \class{QIP}$.
Of these classes, $\class{QIP}(2)$ seems to be the most mysterious.
It is known that $\oplus\class{MIP}^{\ast} \subseteq \class{QIP}(2)$
\cite{Wehner06} and $\class{QSZK}\subseteq\class{QIP}(2)$ 
\cite{Watrous02,Watrous06}.
Here, $\oplus\class{MIP}^{\ast}$ denotes the class of problems having
one-round two-prover classical interactive proof systems in which the
provers share quantum entanglement, answer one bit each, and the
verifier accepts or rejects based on the parity of these bits;
and $\class{QSZK}$ denotes the class of problems having statistical
zero-knowledge quantum interactive proof systems.
No upper bound other than the trivial containment $\class{QIP}(2)
\subseteq \class{QIP}$, which implies
$\class{QIP}(2)\subseteq\class{EXP}$, was previously known.

In this paper we prove the new containment:
\[
\class{QIP}(2) \subseteq \class{PSPACE}.
\]
Similar to $\class{QIP}\subseteq\class{EXP}$, this containment is
proved using semidefinite programming; but this time the containment
is achieved by using an $\class{NC}$ algorithm rather than a
sequential polynomial-time algorithm.
Our algorithm is based on the {\it multiplicative weights update}
method, which was developed by several researchers and is described in
the survey \cite{AroraHK05b} and in the PhD thesis of Kale
\cite{Kale07}.
In particular, our algorithm is based on a general method that was
independently discovered by Arora and Kale \cite{AroraK07} and
Warmuth and Kuzmin \cite{WarmuthK06}.
The key aspect of this approach that makes it useful for proving
$\class{QIP}(2) \subseteq \class{PSPACE}$ is its parallelizability:
it is an iterative method in which each iteration is easily
parallelized, and is such that only a very small number of iterations
is needed for an approximation that is accurate enough for our needs.
A related approach was used by two of us \cite{JainW09} to prove the
containment of a different quantum complexity class (called
$\class{QRG}(1)$) in $\class{PSPACE}$, but the specific technical
details of the simulations are rather different.

The rest of this paper has the following structure.
We begin with Section~\ref{sec:preliminaries}, which includes a brief
discussion of background information needed for the rest of the paper,
including linear algebra notation and parallel algorithms for matrix
computations.
Section~\ref{sec:qip2} introduces two-message quantum interactive
proof systems and establishes a simple fact concerning their
robustness with respect to error bounds.
In Section~\ref{sec:qip2-SDP} we present a semidefinite programming
formulation of the maximum probability with which a verifier in a
two-message quantum interactive proof system can be made to accept,
and the actual simulation of $\class{QIP}(2)$ in $\class{PSPACE}$ is split
into the three sections that follow:
Section~\ref{sec:overview} presents an overview of the simulation,
while Sections~\ref{sec:conditioning} and \ref{sec:Arora-Kale}
describe in more detail its two most technical parts.
The precision requirements of the entire simulation are discussed in
Section~\ref{sec:precision}, and the paper concludes with
Section~\ref{sec:conclusion}.

\section{Preliminaries} \label{sec:preliminaries}

\subsection{Linear algebra notation and terminology}

For complex vector spaces of the form $\X = \complex^N$ and $\Y =
\complex^M$, we write $\lin{\X,\Y}$ to denote the space of linear
operators mapping $\X$ to $\Y$, which is identified with the set of
$M\times N$ complex matrices in the usual way.
An inner product on $\lin{\X,\Y}$ is defined as
$\ip{A}{B} = \tr(A^{\ast} B)$ for all $A,B\in\lin{\X,\Y}$, where
$A^{\ast}$ denotes the adjoint (or conjugate transpose) of $A$.
The notation $\lin{\X}$ is shorthand for $\lin{\X,\X}$, and
the identity operator on $\X$ is denoted $\I_{\X}$ (or just $\I$ when
$\X$ is understood).

\pagebreak[3]

The following special types of operators are relevant to the paper:
\begin{mylist}{\parindent}
\item[1.]
An operator $A\in\lin{\X}$ is {\it Hermitian} if $A = A^{\ast}$.
We write
$\lambda(A) = (\lambda_1(A),\ldots,\lambda_N(A))$
to denote the vector of eigenvalues of $A$, sorted from largest to
smallest:
$\lambda_1(A) \geq \lambda_2(A) \geq \cdots \geq \lambda_N(A)$.

\item[2.]
An operator $P\in\lin{\X}$ is {\it positive semidefinite} if it is
Hermitian and all of its eigenvalues are nonnegative.
The set of such operators is denoted $\pos{\X}$.
The notation $P\geq 0$ also indicates that $P$ is positive
semidefinite, and more generally the notations $A\leq B$ and $B\geq A$
indicate that $B - A\geq 0$ for Hermitian operators $A$ and $B$.

\item[3.]
A positive semidefinite operator $\Pi\in\pos{\X}$ is a
{\it projection} if all of its eigenvalues are either 0 or 1.
(Sometimes such an operator is called an {\it orthogonal projection},
but we have no need to discuss the more general sort of projection.)

\item[4.]
An operator $\rho\in\lin{\X}$ is a {\it density operator} if it is
both positive semidefinite and has trace equal to 1.
The set of such operators is denoted $\density{\X}$.

\item[5.]
An operator $U\in\lin{\X}$ is {\it unitary} if $U^{\ast} U = \I_{\X}$.
The set of such operators is denoted $\unitary{\X}$.
\end{mylist}

Three operator norms are discussed in this paper: the 
{\it trace norm}, {\it Frobenius norm}, and {\it spectral norm},
defined as
\[
\norm{A}_1 = \tr\sqrt{A^{\ast} A}\:,
\quad
\norm{A}_2 = \sqrt{\ip{A}{A}}\:,
\quad\text{and}\quad
\norm{A} = \max\{
\norm{Au}\,:\,u\in\X,\,\norm{u} = 1\}
\]
respectively, for each $A \in \lin{\X}$.
Alternately, these norms are given by the 1, 2 and $\infty$ norms
of the vector of singular values of $A$.
For every operator $A$ it holds that
$\norm{A} \leq \norm{A}_2 \leq \norm{A}_1$.
We also use the inequalities $\abs{\ip{A}{B}}\leq\norm{A}\norm{B}_1$,
$\norm{AB} \leq \norm{A} \norm{B}$, and
$\norm{AB}_1 \leq \norm{A} \norm{B}_1$ a few times in the paper.
The {\it fidelity function} is defined as
\[
\fid(P,Q) = \norm{\sqrt{P}\sqrt{Q}}_1
\]
for positive semidefinite operators $P$ and $Q$ of equal dimension.

A {\it super-operator} is a linear mapping of the form
$\Phi:\lin{\X}\rightarrow\lin{\Y}$, for spaces of the form 
$\X = \complex^N$ and $\Y = \complex^M$.
The identity super-operator on $\lin{\X}$ is denoted $\I_{\lin{\X}}$.
The adjoint super-operator to $\Phi$ is the unique
super-operator $\Phi^{\ast}: \lin{\Y} \rightarrow \lin{\X}$ for which
$\ip{Y}{\Phi(X)}=\ip{\Phi^{\ast}(Y)}{X}$ for all $X\in\lin{\X}$ and
$Y\in\lin{\Y}$.

The following special types of super-operators are relevant to the
paper.
\begin{mylist}{\parindent}
\item[1.]
  $\Phi:\lin{\X}\rightarrow\lin{\Y}$ is {\it completely positive}
  if it holds that
  $(\Phi\otimes\I_{\lin{\W}})(P) \in \pos{\Y\otimes\W}$
  for every choice of $\W = \complex^k$ and $P\in\pos{\X\otimes\W}$.
  
\item[2.]
  $\Phi:\lin{\X}\rightarrow\lin{\Y}$ is {\it trace-preserving} if
  $\tr(\Phi(X)) = \tr(X)$
  for every $X\in\lin{\X}$.

\item[3.]
  $\Phi:\lin{\X}\rightarrow\lin{\Y}$ is a {\it quantum operation}
  (also called an {\it admissible super-operator} or a {\it quantum
  channel}) if it is both completely positive and trace-preserving.

\end{mylist}

\subsection{Remarks on NC and parallel matrix computations}
\label{sec:NC}

To prove that $\class{QIP}(2)$ is contained in $\class{PSPACE}$, we
will make use of various facts concerning parallel computation.
First, let us recall the definition of two complexity classes based
on bounded-depth circuit families: $\class{NC}$ and
$\class{NC}(\mathit{poly})$.
The class $\class{NC}$ contains all functions that can be computed by
logarithmic-space uniform Boolean circuits of polylogarthmic depth,
while the class $\class{NC}(\mathit{poly})$ contains all functions
that can be computed by polynomial-space uniform families of Boolean
circuits having polynomial-depth.
By restricting these classes to predicates we obtain classes of
languages (or more generally promise problems).

There are a few facts about these classes that we will need.
The first fact, which follows from \cite{Borodin77}, is that for
languages (or promise problems) we have
$\class{NC}(\mathit{poly})\subseteq\class{PSPACE}$.
(In fact it holds that $\class{NC}(\mathit{poly}) = \class{PSPACE}$,
but we only need a containment in one direction.)
The second fact is that functions in these classes compose nicely.
In particular, if $F:\{0,1\}^{\ast} \rightarrow \{0,1\}^{\ast}$ is a
function in $\class{NC}(\mathit{poly})$ and
$G:\{0,1\}^{\ast} \rightarrow \{0,1\}^{\ast}$ is a function in
$\class{NC}$, then the composition $G\circ F$ is also in
$\class{NC}(\mathit{poly})$.
This follows from the most obvious way of composing the families of
circuits that compute $F$ and $G$, along with the observation that
$\abs{F(x)}$ can be at most exponential in $\abs{x}$.

Finally, we will make use of the fact that many computations involving
matrices can be performed by $\class{NC}$ algorithms.
We will restrict our attention to matrix computations on matrices whose
entries have rational real and imaginary parts.
Numbers of this form, $\alpha = (a/b) + i (c/d)$ for integers $a$, $b$,
$c$, and $d$, are sometimes referred to as Gaussian rationals.
We assume any number of this form is encoded as a 4-tuple $(a,b,c,d)$
using binary notation, so that the {\it length} of $\alpha$ is
understood to be the total number of bits needed for such an
encoding.

It is known that elementary matrix operations, such as additions,
multiplications, and inversions, can be performed in $\class{NC}$.
(The survey \cite{vzGathen93}, for instance, describes $\class{NC}$
algorithms for these tasks.)
We will also make use of the fact that matrix exponentials
and spectral decompositions can be approximated to high precision in
$\class{NC}$.
In more precise terms, we have that the following problems are in
$\class{NC}$:

\begin{center}
\underline{Matrix exponentials}\\[2mm]
\begin{tabular}{lp{5.5in}}
{\it Input:} &
An $n\times n$ matrix $M$, a positive rational number $\varepsilon$,
and an integer $k$ expressed in unary notation (i.e., $1^k$),
such that $\norm{M} \leq k$.\\[1mm]
{\it Output:} &
An $n\times n$ matrix $X$ such that 
$\norm{\exp(M) - X} < \varepsilon$.
\end{tabular}
\end{center}

\vspace{2mm}

\begin{center}
\underline{Spectral decompositions}\\[2mm]
\begin{tabular}{lp{5.5in}}
{\it Input:} &
An $n\times n$ Hermitian matrix $H$ and a positive rational number
$\varepsilon$.\\[1mm]
{\it Output:} &
An $n\times n$ unitary matrix $U$ and an $n\times n$ real diagonal matrix
$\Lambda$ such that
\[
\norm{M - U \Lambda U^{\ast}} < \varepsilon.
\]
\end{tabular}
\end{center}

\vspace{-8mm}

\begin{center}
\underline{Singular-value decompositions}\\[2mm]
\begin{tabular}{lp{5.5in}}
{\it Input:} &
An $n\times m$ matrix $M$ and a positive rational number
$\varepsilon$.\\[1mm]
{\it Output:} &
An $n\times r$ matrix $U$ with orthonormal columns,
an $m\times r$ matrix $V$ with orthonormal columns,
and an $r\times r$ diagonal matrix $\Sigma$ with positive diagonal entries
such that
\[
\norm{M - U \Sigma V^{\ast}} < \varepsilon.
\]
\end{tabular}
\end{center}

\vspace{-4mm}

\noindent
Note that in these problems, the description of $\varepsilon$ has
roughly $\log(1/\varepsilon)$ bits, which means that highly accurate
approximations are possible in $\class{NC}$.
The fact that matrix exponentials can be approximated in
$\class{NC}$ as claimed follows by truncating the series
\[
\exp(M) = \I + M + M^2/2 + M^3/6 + \cdots
\]
to a number of terms polynomial in $k$ and $\log(1/\varepsilon)$.
(This is not a very practial way to compute matrix exponentials, but
it establishes the fact we need.)
The fact that spectral and singular value decompositions can be
approximated in $\class{NC}$ follows from a composition of known
facts: in $\class{NC}$ one can compute characteristic polynomials and
null spaces of matrices, perform orthogonalizations of vectors, and
approximate roots of integer polynomials to high precision
\cite{Csanky76,BorodinGH82,BorodinCP83,BenOrFKT86,vzGathen93,Neff94}.

\section{Two-message quantum interactive proof systems}
\label{sec:qip2}

The purpose of this section is to introduce the class
$\class{QIP}(2)$, including its definition and a simple proof that it
is robust with respect to error bounds.
For a general discussion of quantum interactive proof systems, as
opposed to the somewhat simplified case in which only two messages are
exchanged, the reader is referred to \cite{KitaevW00} and
\cite{Watrous09-complexity}.

\subsection{Definition of two-message quantum interactive proofs}

To define the class $\class{QIP}(2)$, we begin by defining a
{\it two-message quantum verifier $V$} as a classical polynomial-time
algorithm that, on each input string $x$, outputs the
description of two quantum circuits: $U_x$ and $V_x$.
The circuit $U_x$ describes the verifier's initial preparation of a
state, part of which is sent to the prover, while the circuit $V_x$
describes the verifier's actions upon receiving a response from the
prover.
For the sake of simplicity, and without loss of generality, we
assume that for every input string $x$, the circuits $U_x$ and
$V_x$ are both composed of gates from some finite, universal set of
unitary quantum gates whose entries have rational real and imaginary
parts.
The number of qubits on which the circuits $U_x$ and $V_x$ act is 
assumed to be equal to $2 p(n)$, where $n = \abs{x}$ and $p$ is some 
polynomial-bounded function.
The first $p(n)$ qubits represent the communication channel between
the prover and verifier, while the remaining $p(n)$ qubits serve as
the private memory of the verifier.
(It is not really necessary that the number of message qubits and
private memory qubits agree, but it causes no change in the
computational power of the model.)

A two-message quantum prover $P$ is simply a collection of quantum
operations (or, equivalently, completely positive and trace preserving
super-operators)
$\{\Psi_x\,:\,x\in\{0,1\}^{\ast}\}$.
Such a prover is compatible with a given verifier $V$ if each
operation $\Psi_x$ acts on $p(n)$ qubits for the function $p$
mentioned above.

An interaction between a two-message verifier $V$ and a compatible
prover $P$ on an input $x$ proceeds as follows:
\begin{mylist}{\parindent}
\item[1.] $2p(n)$ qubits are initialized in the $\ket{0}$ state. 
\item[2.] The circuit $U_x$ is applied to all of the qubits.
\item[3.] The prover's operation $\Psi_x$ is applied to the first
  $p(n)$ qubits.
\item[4.] The circuit $V_x$ is applied to all of the qubits.
\item[5.] The first qubit is measured with respect to the standard
  basis, with the outcome 1 indicating {\it acceptance} and 0
  indicating {\it rejection}.
\end{mylist}
Figure~\ref{fig:qip2} illustrates such an interaction.
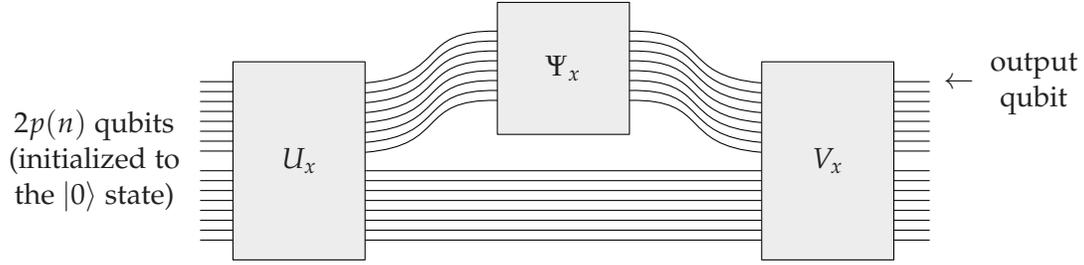
\begin{figure}[t]
  \begin{center}
    \unitlength=0.25pt
    \begin{picture}(1000, 410)(-28,50)
      \gasset{Nmr=0,AHnb=0}
      \node[Nw=200,Nh=300,fillgray=0.93](V1)(100,200){$U_x$}
      \node[Nw=200,Nh=300,fillgray=0.93](V2)(900,200){$V_x$}
      \node[Nframe=n,Nh=500,Nw=320](Nleft)(-210,200){%
        \begin{tabular}{c}
          $2p(n)$ qubits\\
          (initialized to\\
          the $\ket{0}$ state)
      \end{tabular}}

      \node[Nframe=n,Nh=500](Nright)(1100,200){}
      
      \drawedge[syo=-15,eyo=-15](V1,V2){}
      \drawedge[syo=-30,eyo=-30](V1,V2){}
      \drawedge[syo=-45,eyo=-45](V1,V2){}
      \drawedge[syo=-60,eyo=-60](V1,V2){}
      \drawedge[syo=-75,eyo=-75](V1,V2){}
      \drawedge[syo=-90,eyo=-90](V1,V2){}
      \drawedge[syo=-105,eyo=-105](V1,V2){}
      \drawedge[syo=-120,eyo=-120](V1,V2){}

      \drawedge[syo=-15,eyo=-15](V1,Nleft){}
      \drawedge[syo=-30,eyo=-30](V1,Nleft){}
      \drawedge[syo=-45,eyo=-45](V1,Nleft){}
      \drawedge[syo=-60,eyo=-60](V1,Nleft){}
      \drawedge[syo=-75,eyo=-75](V1,Nleft){}
      \drawedge[syo=-90,eyo=-90](V1,Nleft){}
      \drawedge[syo=-105,eyo=-105](V1,Nleft){}
      \drawedge[syo=-120,eyo=-120](V1,Nleft){}

      \drawedge[syo=-15,eyo=-15](V2,Nright){}
      \drawedge[syo=-30,eyo=-30](V2,Nright){}
      \drawedge[syo=-45,eyo=-45](V2,Nright){}
      \drawedge[syo=-60,eyo=-60](V2,Nright){}
      \drawedge[syo=-75,eyo=-75](V2,Nright){}
      \drawedge[syo=-90,eyo=-90](V2,Nright){}
      \drawedge[syo=-105,eyo=-105](V2,Nright){}
      \drawedge[syo=-120,eyo=-120](V2,Nright){}

      \drawedge[syo=15,eyo=15](V1,Nleft){}
      \drawedge[syo=30,eyo=30](V1,Nleft){}
      \drawedge[syo=45,eyo=45](V1,Nleft){}
      \drawedge[syo=60,eyo=60](V1,Nleft){}
      \drawedge[syo=75,eyo=75](V1,Nleft){}
      \drawedge[syo=90,eyo=90](V1,Nleft){}
      \drawedge[syo=105,eyo=105](V1,Nleft){}
      \drawedge[syo=120,eyo=120](V1,Nleft){}

      \drawedge[syo=15,eyo=15](V2,Nright){}
      \drawedge[syo=30,eyo=30](V2,Nright){}
      \drawedge[syo=45,eyo=45](V2,Nright){}
      \drawedge[syo=60,eyo=60](V2,Nright){}
      \drawedge[syo=75,eyo=75](V2,Nright){}
      \drawedge[syo=90,eyo=90](V2,Nright){}
      \drawedge[syo=105,eyo=105](V2,Nright){}
      \drawedge[syo=120,eyo=120](V2,Nright){}

      \node[Nw=200,Nh=200,fillgray=0.93](P)(500,340){$\Psi_x$}
     
      \drawbpedge[syo=15,eyo=-60](V1,-5,300,P,170,300){}
      \drawbpedge[syo=30,eyo=-45](V1,-5,300,P,170,300){}
      \drawbpedge[syo=45,eyo=-30](V1,-5,300,P,170,300){}
      \drawbpedge[syo=60,eyo=-15](V1,-5,300,P,170,300){}
      \drawbpedge[syo=75,eyo=0](V1,-5,300,P,170,300){}
      \drawbpedge[syo=90,eyo=15](V1,-5,300,P,170,300){}
      \drawbpedge[syo=105,eyo=30](V1,-5,300,P,170,300){}
      \drawbpedge[syo=120,eyo=45](V1,-5,300,P,170,300){}

      \drawbpedge[syo=15,eyo=-60](V2,185,300,P,10,300){}
      \drawbpedge[syo=30,eyo=-45](V2,185,300,P,10,300){}
      \drawbpedge[syo=45,eyo=-30](V2,185,300,P,10,300){}
      \drawbpedge[syo=60,eyo=-15](V2,185,300,P,10,300){}
      \drawbpedge[syo=75,eyo=0](V2,185,300,P,10,300){}
      \drawbpedge[syo=90,eyo=15](V2,185,300,P,10,300){}
      \drawbpedge[syo=105,eyo=30](V2,185,300,P,10,300){}
      \drawbpedge[syo=120,eyo=45](V2,185,300,P,10,300){}

      \put(1075,310){$\leftarrow$%
        \begin{tabular}{c}
          output\\qubit
        \end{tabular}
      }

    \end{picture}
  \end{center}
  \caption{An interaction between a verifier $V$ and a prover $P$ on
    an input $x$.
    The verifier's actions are determined by the circuits $U_x$ and
    $V_x$ acting on $2p(n)$ qubits, while the prover's action
    corresponds to the quantum operation $\Psi_x$ on just the
    first $p(n)$ qubits.}
  \label{fig:qip2}
\end{figure}

Now, a promise problem $A = (A_{\yes},A_{\no})$ is in $\class{QIP}(2)$
if and only if there exists a two-message verifier $V$ with the
following completeness and soundness properties:
\begin{mylist}{\parindent}
\item[1.] (Completeness)
If $x\in A_{\yes}$, then there exists a prover $P$ that causes $V$ to
accept $x$ with probability at least 2/3.

\item[2.] (Soundness)
If $x\in A_{\no}$, then every prover $P$ that is compatible with $V$
causes $V$ to accept $x$ with probability at most 1/3.
\end{mylist}

\subsection{Robustness of QIP(2) with respect to error bounds}

It was proved in \cite{KitaevW00} that quantum interactive proof
systems with negligible completeness error are amenable to parallel
repetition.
This allows for an exponential reduction in error for quantum
interactive proof systems with three or more messages, because such
proof systems can be transformed to have perfect completeness by a
different method.
However, this method does not work for two-message quantum interactive
proof systems, because the perfect-completeness transformation
requires the addition of messages.
So, we will require a different method of error reduction.

Assume that $A$ is a promise problem in $\class{QIP}(2)$ and that
$(V,P)$ is a two-message quantum interactive proof system for $A$ with
completeness and soundness probabilities bounded by $a$ and $b$, where
$a - b \geq 1/q$ for some polynomial-bounded function $q$.
We will define a new verifier $V'$ that has completeness probability at least
$1 - 2^{-r}$ and soundness probability at most $2^{-r}$, for any
choice of a polynomial-bounded function $r$.
A description of $V'$ follows.

\begin{mylist}{5mm}
\item[1.]
Let $s = 2 r q$ and let $t = 8 r q^2 s$.
Run $s t$ independent, parallel executions of the protocol for $V$,
one for each pair $(i,j)$ with $i\in\{1,\ldots,s\}$ and $j\in\{1,\ldots,t\}$.
Measure the output qubit for each execution, and let the result of the
measurement for execution $(i,j)$ be $y_{i,j}\in\{0,1\}$.

\item[2.]
For each $i = 1,\ldots,s$, set
\[
z_i = \left\{\begin{array}{ll}
1 & \text{if } \sum_{j=1}^t y_{i,j} \geq t\cdot \frac{a+b}{2}\\[2mm]
0 & \text{otherwise}
\end{array}\right.
\]

\item[3.]
Accept if $\bigwedge_{i=1}^s z_i = 1$, reject otherwise.
\end{mylist}

Now let us consider the maximum probability with which $V'$ can be made to
accept.
Suppose first that an input $x\in A_{\mathrm{yes}}$ is fixed, so that $V$ is
made to accept with probability at least $a$ by the prover $P$.
Our goal is to define a prover $P'$ that causes $V'$ to accept with
probability at least $1 - 2^{-r}$.
This is easily done by defining $P'$ so that it runs $s t$ independent
simulations of $P$.
Let $Y_{i,j}$ and $Z_i$, for $i\in\{1,\ldots,s\}$ and
$j\in\{1,\ldots,t\}$, be Boolean-valued random variables corresponding to the
values $y_{i,j}$ and $z_i$ when $V'$ interacts with the prover $P'$ just 
described.
Given that $P'$ simulates $s t$ independent copies of $P$, we have that the
random variables $Y_{i,j}$ are independent and satisfy $\op{E}[Y_{i,j}]\geq a$
for each pair $(i,j)$.
By the Chernoff Bound, we therefore have
\[
\op{Pr}[Z_i = 0] = \op{Pr}\left[Y_{i,1} + \cdots + Y_{i,t} < a t
\(1 - \frac{a-b}{2a}\)\right]
\leq \exp\(-\frac{t}{8 a}(a-b)^2\)
\leq e^{-rs},
\]
and thus the probability of rejection is at most
$s e^{-r s} <  2^{-r}$.

Suppose on the other hand that an input $x\in A_{\mathrm{no}}$ is fixed, so
that no prover $P$ can convince $V$ to accept with probability greater than
$b$.
Fix an arbitrary prover $P'$, and as before let $Y_{i,j}$ and $Z_i$ be
Boolean-valued random variables corresponding to $y_{i,j}$ and $z_i$.
There may not be independence among these random variables, as $P'$ may not
treat the parallel executions independently.
We do know, however, that $\op{E}[Y_{i,j}] \leq b$ for every $i,j$,
given that the maximum acceptance probability of $V$ is $b$.
By Markov's Inequality we therefore have
\[
\op{Pr}[Z_i = 1] = \op{Pr}\left[Y_{i,1} + \cdots + Y_{i,t} \geq
  \frac{t(a+b)}{2}\right]
\leq 2 \frac{\op{E}[Y_{i,1} + \cdots + Y_{i,t}]}{t(a+b)}
< 1 - \frac{1}{2 q}.
\]
We may view $z_1,\ldots,z_s$ as being the outcomes of $s$ parallel
executions of a quantum interactive proof system that accepts with probability
at most $1 - \frac{1}{2q}$.
The verifier $V'$ accepts if and only if all of these executions accept,
and so by the result on parallel repetition proved in \cite{KitaevW00}
we may conclude that the probability of acceptance of $V'$ is at most
\[
\(1 - \frac{1}{2q}\)^s < \exp\(-\frac{s}{2q}\)
< 2^{-r}.
\]

Thus, the verifier $V'$ has been shown to have completeness and soundness
probabilities as required, completing the proof.

\section{Maximum acceptance probability as a semidefinite program}
\label{sec:qip2-SDP}

The maximum acceptance probability of a verifier in a quantum
interactive proof system can be phrased as semidefinite programming
problem \cite{KitaevW00}.
For this paper a semidefinite programming formulation based on ones
described in \cite{GutoskiW07,Watrous09-cb-sdp} will be used.

Suppose $V$ is a two-message verifier, and that an input string $x$ of
length $n$ is being considered.
Let us also define
\[
\ket{\psi} = U_x \ket{0^{2p(n)}}
\quad\quad\text{and}\quad\quad
\Pi = V_x^{\ast} (\ket{1}\!\bra{1} \otimes \I) V_x.
\]
In words, $\ket{\psi}$ denotes the quantum state initially prepared by
the verifier, the first half of which is sent to the prover;
and $\Pi$ denotes the projection operator corresponding to the 
{\it accept} outcome of the measurement that the verifier effectively
performs after receiving the prover's message.

For convenience, we will assign distinct names to the complex vector
spaces that arise from an interaction between $V$ on input $x$ and a
compatible prover operation $\Psi$.
Specifically, let $\X$ denote the space corresponding to the
verifier's message to the prover, let $\Y$ denote the space
corresponding to the prover's response, and let $\Z$ denote the space
corresponding to the verifier's private qubits.
Thus, it holds that $\ket{\psi} \in \X\otimes\Z$ and $\Pi$ is a
projection on $\Y\otimes\Z$.
When the prover applies the operation
$\Psi:\lin{\X}\rightarrow\lin{\Y}$, the verifier accepts with
probability
\begin{equation}\label{eq:prob-accept}
\ip{\Pi}{(\Psi\otimes\I_{\lin{\Z}})(\ket{\psi}\!\bra{\psi})}.
\end{equation}

To express the maximum probability for $V$ to accept, over all choices
of an operation $\Psi$, as a semidefinite program, it is helpful to
recall the {\it Choi-Jamio{\l}kowski representation} of
super-operators.
Let us take $\{\ket{0},\ldots,\ket{N-1}\}$ to be
the standard basis of $\X$.
Then the Choi-Jamio{\l}kowski representation of
$\Psi:\lin{\X}\rightarrow\lin{\Y}$ is the operator
$J(\Psi)\in\lin{\Y\otimes\X}$ defined by
\[
J(\Psi) = \sum_{0 \leq i,j\leq N-1}
\Psi(\ket{i}\!\bra{j}) \otimes \ket{i}\!\bra{j}.
\]
It holds that $\Psi$ is completely positive if and only if $J(\Psi)$
is positive semidefinite \cite{Jamiolkowski72,Choi75}, and $\Psi$ is
trace-preserving if and only if $\tr_{\Y}(J(\Psi)) = \I_{\X}$.

Now let us write
\[
\ket{\psi} = \sum_{0\leq i \leq N-1} \ket{i} \ket{\psi_i}
\]
for vectors $\ket{\psi_0},\ldots,\ket{\psi_{N-1}}\in\Z$, and define
$B\in\lin{\X,\Z}$ as
\[
B = \sum_{0\leq i \leq N-1} \ket{\psi_i}\bra{i}.
\]
Then it is clear that
\[
(\Psi\otimes\I_{\lin{\Z}})(\ket{\psi}\!\bra{\psi})
= (\I_{\Y} \otimes B) J(\Psi) (\I_{\Y} \otimes B^{\ast}).
\]
We therefore find that the probability of acceptance
\eqref{eq:prob-accept} may alternately be written
\[
\ip{\Pi}{(\I_{\Y} \otimes B)J(\Psi)(\I_{\Y} \otimes B^{\ast})}
= \ip{(\I_{\Y} \otimes B^{\ast})\Pi(\I_{\Y}\otimes B)}{J(\Psi)}
= \ip{Q}{J(\Psi)}
\]
for $Q = (\I_{\Y} \otimes B^{\ast})\Pi(\I_{\Y}\otimes B)$.
We call $Q$ the {\it interactive measurement operator} that is
determined by $V$ on input $x$.
It is clear that the interactive measurement operator
$Q$ is positive semidefinite, and moreover that
$Q\leq \I_{\Y} \otimes \xi$ for the density operator 
$\xi = B^{\ast}B$.

Now, let us define
\[
\mu(Q) = \max_{\Psi} \ip{Q}{J(\Psi)},
\]
where the maximum is over all valid quantum operations of the form
$\Psi:\lin{\X}\rightarrow\lin{\Y}$.
The quantity $\mu(Q)$ will be called the 
{\it maximum acceptance probability} of $Q$, as this value is
precisely the maximum acceptance probability of the verifier $V$ on
input $x$, whose description alone has led us to the definition of
$Q$.
As stated above, when $\Psi:\lin{\X}\rightarrow\lin{\Y}$ ranges over
the set of all valid quantum operations, $J(\Psi)$ ranges over the set
of positive semidefinite operators satisfying the linear constraint
$\tr_{\Y}(J(\Psi))=\I_{\X}$.
This implies that the quantity $\mu(Q)$ is represented by a
semidefinite program:
\begin{align*}
  \text{maximize:}\quad & \ip{Q}{X}\\
  \text{subject to:}\quad & 
  \tr_{\Y}(X) \leq \I_{\X},\\
  & X\in\pos{\Y\otimes\X}.
\end{align*}
The feasible region of this semidefinite program is well-bounded
(in the sense of \cite{GrotschelLS93}), and therefore its optimal
value $\mu(Q)$ can be approximated to high precision in time
polynomial in the size of $Q$ (which is exponential in $\abs{x}$).
This fact does not help us to prove
$\class{QIP}(2)\subseteq\class{PSPACE}$, however.
As is described in the next section, we will need an $\class{NC}$
algorithm rather than just a polynomial-time algorithm to draw this
conclusion.

It will be necessary for us to rephrase the semidefinite program
above, and to explicitly state its dual program.
As will be discussed shortly, we will only need to consider this
formulation for invertible interactive measurement operators, so $Q$
is hereafter assumed to be invertible.
Define a super-operator $\Phi:\lin{\Y\otimes\X}\rightarrow\lin{\X}$ as
\[
\Phi(X) = \tr_{\Y}\(Q^{-1/2} X Q^{-1/2}\).
\]
The adjoint super-operator
$\Phi^{\ast}:\lin{\X}\rightarrow\lin{\Y\otimes\X}$ to $\Phi$ is given
by
\[
\Phi^{\ast}(Y) = Q^{-1/2} ( \I_{\Y} \otimes Y ) Q^{-1/2}.
\]
The value $\mu(Q)$ is then seen to be the optimal value of the
following semidefinite program:
\begin{center}
  \begin{minipage}{3in}
    \centerline{\underline{Primal problem}}\vspace{-7mm}
    \begin{align*}
    \text{maximize:}\quad & \tr(X)\\
    \text{subject to:}\quad & 
    \Phi(X) \leq \I_{\X},\\
    & X\in\pos{\Y\otimes\X}.
    \end{align*}
  \end{minipage}
  \begin{minipage}{3in}
    \centerline{\underline{Dual problem}}\vspace{-7mm}
    \begin{align*}
    \text{minimize:}\quad & \tr(Y) \\
    \text{subject to:}\quad 
    & \Phi^{\ast}(Y) \geq \I_{\Y\otimes\X}, \\
    & Y \in \pos{\X}.
    \end{align*}
  \end{minipage}
\end{center}
Strong duality follows from strict feasibility, which is easily
verified, and so the optimal primal and dual solutions are given by
$\mu(Q)$.

\section{Overview of the simulation} \label{sec:overview}

We will now explain, in high-level terms, our simulation of
$\class{QIP}(2)$ in $\class{PSPACE}$.
To prove that $\class{QIP}(2)\subseteq\class{PSPACE}$, it will
suffice to prove $\class{QIP}(2) \subseteq \class{NC}(\mathit{poly})$.
This will be facilitated by the fact, discussed in
Section~\ref{sec:NC}, that many computations involving matrices,
including elementary operations such as addition, multiplication, and
inversion, as well as approximations of spectral decompositions,
singular-value decompositions, and matrix exponentials, can be
performed in~$\class{NC}$.

For the remainder of this paper, assume that $A = (A_{\yes},A_{\no})$
is an arbitrary promise problem in $\class{QIP}(2)$, and let $V$ be a
two-message verifier for $A$ that has exponentially small completeness
and soundness error.
The goal of the simulation is to determine whether or not $V$ can be
made to accept a given input string $x$ with high probability.
The variable $n$ will always denote the input length $n = \abs{x}$,
and $p(n)$ will denote the number of qubits exchanged by the verifier
and prover on each of the two messages as discussed in
Section~\ref{sec:qip2}.

There are three main steps of the simulation:
\begin{mylist}{\parindent}
\item[1.]
Compute from $x$ an explicit description of
$\ket{\psi}$ and $\Pi$.

\item[2.]
Process the description of the vector $\ket{\psi}$ and the projection
$\Pi$ into a well-conditioned interactive measurement operator $Q$ and
positive rational numbers $\gamma$ and $\varepsilon$ satisfying
\begin{align*}
x\in A_{\yes} & \;\Rightarrow\;\mu(Q) \geq (1 + 4\varepsilon)\gamma,\\
x\in A_{\no} & \;\Rightarrow\;\mu(Q)\; \leq (1 - 4\varepsilon)\gamma.
\end{align*}
For some polynomial $q$ it will hold that $\kappa(Q) \leq q(n)$,
where $\kappa(Q) = \norm{Q}\snorm{Q^{-1}}$ denotes the
condition number of $Q$.
Moreover it will hold that $1/q(n) \leq \varepsilon$ and
$1/q(n) \leq \gamma$.

\item[3.]
Use a parallel algorithm, based on the multiplicative weights update
method, to test whether $\mu(Q)$ is larger or smaller than
$\gamma$.
\end{mylist}

\noindent
The first step is easily performed in $\class{NC}(\mathit{poly})$,
using an exact computation.
In particular, one may simply compute products of the matrices that
describe the individual gates of the verifier's circuits.
Given that this step is straightforward, we will not comment on it
further.
The second and third steps are more complicated, and are
described separately in Sections~\ref{sec:conditioning} and
\ref{sec:Arora-Kale} below.
Both correspond to $\class{NC}$ computations (where the input size
is exponential in $n$), and by composing these computations with the
first step just described, we will obtain that $A$ is in
$\class{NC}(\mathit{poly})$, and therefore
$\class{QIP}(2)\subseteq\class{PSPACE}$.

\section{Preparing a well-conditioned interactive measurement
  operator} \label{sec:conditioning}

After the first step of the simulation, we have a unit vector
$\ket{\psi}$ and a projection operator $\Pi$.
Let us write $M = 2^{p(n)}$ to denote the dimension of both of the
message spaces and the verifier's private work space defined by $V$ on
input $x$, and let us also define $\X_0=\complex^M$, $\Y=\complex^M$,
and $\Z_0=\complex^M$.
We view that the space $\X_0$ corresponds to the verifier's message to
the prover, that $\Y$ corresponds to the prover's message back to the
verifier, and that $\Z_0$ represents the verifier's private workspace;
and thus $\ket{\psi}\in \X_0\otimes\Z_0$ and
$\Pi\in\pos{\Y\otimes\Z_0}$.
The reason for subscripting $\X_0$ and $\Z_0$ with 0 is that we are
viewing these as initial choices of spaces.
In the processing of $\ket{\psi}$ and $\Pi$, we will define an
interactive measurement operator $Q$ over spaces $\X$, $\Y$ and $\Z$
where $\X = \complex^N$ and $\Z = \complex^N$ for some choice of a
positive integer $N\leq M$.

Along the same lines as was discussed in Section~\ref{sec:qip2}, we
may define an interactive measurement operator
$R\in\pos{\Y\otimes\X_0}$ as
$R = (\I_{\Y} \otimes B^{\ast})\Pi (\I_{\Y}\otimes B)$,
for
\[
\ket{\psi} = \sum_{0\leq i \leq M-1} \ket{i} \ket{\psi_i}
\quad\quad\text{and}\quad\quad
B = \sum_{0\leq i \leq M-1} \ket{\psi_i} \bra{i}.
\]
The quantity $\mu(R)$ is precisely the maximum acceptance probability
of $V$ on input $x$, but nothing can be said about the condition
number of $R$ (which may not even be invertible).

Our goal is to compute a new measurement operator
$Q\in\pos{\Y\otimes\X}$, where $\X = \complex^N$ for some choice of
$N\leq M$, along with positive rational numbers $\gamma$ and
$\varepsilon$, such that the following properties hold for some
polynomial $q$:
\begin{mylist}{\parindent}
\item[1.] The interactive measurement operator $Q$ is
  well-conditioned: $\kappa(Q) \leq q(n)$.
\item[2.] The values $\gamma$ and $\varepsilon$ are non-negligible:
$1/q(n) \leq \varepsilon$ and $1/q(n) \leq \gamma$.
\item[3.] The value $\mu(Q)$ satisfies the properties
  \begin{equation} \label{eq:accept-reject}
    \begin{split}
      x\in A_{\yes} & \;\Rightarrow\;\mu(Q) \geq (1 + 4\varepsilon)\gamma,\\
      x\in A_{\no} & \;\Rightarrow\;\mu(Q)\; \leq (1 - 4\varepsilon)\gamma.
    \end{split}
  \end{equation}
\end{mylist}

The first step in this process is to replace $\ket{\psi}$ by a more
uniform vector $\ket{\phi}\in\X_0\otimes\Z_0$ that is 
``similar enough'' to $\ket{\psi}$ in the sense to be described.
We will, in particular, take $\ket{\phi}$ to be maximally entangled
over certain subspaces of $\X_0$ and $\Z_0$.
This is done by performing the following operations:
\begin{mylist}{\parindent}
\item[1.]
Let
\[
\ket{\psi} = \sum_{0\leq j \leq M-1} \sqrt{\lambda_j} \ket{x_j}\ket{z_j}
\]
be a Schmidt decomposition of $\ket{\psi}$.

\item[2.]
For each positive integer $i$, define the interval
$I_i = \( 2^{-i}, 2^{-i+1}\right]$,
and define 
\[
\Sigma_i=\left\{j\in\{0,\ldots,M-1\}\,:\,\lambda_j\in I_i\right\}.
\]

\item[3.]
Let $k = p(n) + 1$ and choose $i\in\{1,\ldots,k\}$ so that
\[
\sum_{j \in \Sigma_i} \lambda_j \geq \frac{1}{2k}.
\]
The fact that such an $i$ exists is proved below, and hereafter we
write $\Sigma = \Sigma_i$ for this choice of~$i$.

\item[4.]
Define
\[
\ket{\phi} = \frac{1}{\sqrt{\abs{\Sigma}}} \sum_{j\in \Sigma}
\ket{x_j}\ket{z_j}.
\]
\end{mylist}

Now, consider the interactive measurement operator
$S\in\pos{\Y\otimes\X_0}$ that is obtained by replacing $\ket{\psi}$
with $\ket{\phi}$ (with $\Pi$ unchanged).
In other words, $S$ is defined by the same process as $R$ (which was
determined by $\ket{\psi}$ and $\Pi$ as described above), and
satisfies the equation
\[
\ip{S}{J(\Psi)} =
\ip{\Pi}{(\Psi\otimes\I_{\lin{\Z_0}})(\ket{\phi}\!\bra{\phi})}
\]
for every super-operator $\Psi:\lin{\X_0}\rightarrow\lin{\Y}$.
We will prove that
\begin{equation} \label{eq:bounds-on-S}
\mu(R) - \( 1 - \frac{1}{8k}\)
\leq \mu(S) \leq 4k\,\mu(R).
\end{equation}
These are fairly loose bounds---but for the two extremes where
$\mu(R)$ is exponentially close to 0 or~1, the corresponding values for
$\mu(S)$ will be separated by the reciprocal of a polynomial, which is
good enough for our needs.

First, note that
\[
\sum_{i>k} \sum_{j\in \Sigma_i} \lambda_j \leq M 2^{-k} \leq \frac{1}{2}
\]
and therefore
\[
\sum_{i = 1}^k \sum_{j\in \Sigma_i} \lambda_j \geq \frac{1}{2}.
\]
Thus, there must exist a suitable choice of $i$ in step 3 so that
\[
\sum_{j\in \Sigma_i} \lambda_j \geq \frac{1}{2k},
\]
as claimed.
A lower bound on the size of $\Sigma = \Sigma_i$ may be obtained by
noting that
\[
\sum_{j\in \Sigma}\lambda_j \leq 2^{-i+1}\abs{\Sigma}
\]
and therefore
\[
\abs{\Sigma} \geq \frac{2^i}{4k}.
\]

Now, the inner product of $\ket{\psi}$ and $\ket{\phi}$ is easily
bounded from below as
\[
\braket{\phi | \psi} = 
\sum_{j\in \Sigma} \sqrt{\frac{\lambda_j}{\abs{\Sigma}}}
\geq \abs{\Sigma} \sqrt{\frac{2^{-i}}{\abs{\Sigma}}}
= \sqrt{\abs{\Sigma} 2^{-i}}
\geq \frac{1}{\sqrt{4k}},
\]
and therefore
\[
\norm{
\ket{\psi}\!\bra{\psi} - \ket{\phi}\!\bra{\phi}}_1 = 
2 \sqrt{1 - \abs{\braket{\phi | \psi}}^2}
\leq 2 \sqrt{1 - \frac{1}{4k}} \leq 2 \( 1 - \frac{1}{8k} \).
\]
For every choice of an admissible super-operator
$\Psi:\lin{\X_0}\rightarrow\lin{\Y}$ it holds that
\[
0\leq\(\Psi^{\ast}\otimes\I_{\lin{\Z_0}}\)(\Pi)\leq \I_{\X_0\otimes\Z_0},
\]
and by combining this observation with the fact that
$\ket{\psi}\!\bra{\psi} - \ket{\phi}\!\bra{\phi}$ is traceless we
obtain
\[
\ip{R - S}{J(\Psi)}
= \ip{\(\Psi^{\ast}\otimes\I_{\lin{\Z_0}}\)(\Pi)}{
\ket{\psi}\!\bra{\psi} - \ket{\phi}\!\bra{\phi}}
\leq \frac{1}{2} \norm{ \ket{\psi}\!\bra{\psi} - 
\ket{\phi}\!\bra{\phi}}_1 \leq 1 - \frac{1}{8k}.
\]
Therefore
\[
\mu(R) - \mu(S) \leq 1 - \frac{1}{8k},
\]
which establishes the lower bound on $\mu(S)$ claimed in
\eqref{eq:bounds-on-S} above.

To establish the upper bound on $\mu(S)$, which is the second
inequality in \eqref{eq:bounds-on-S}, let us first choose an
admissible super-operator $\Psi$ so that
\[
\mu(S) = \ip{\Pi}{\(\Psi\otimes\I_{\lin{\Z_0}}\)(\ket{\phi}\!\bra{\phi})}.
\]
Now, observe that
\[
\frac{1}{4k} \frac{1}{\abs{\Sigma}} \leq 2^{-i} < \lambda_j
\]
for each $j\in \Sigma$, and thus
\[
\frac{1}{4k} \tr_{\X_0}(\ket{\phi}\!\bra{\phi}) \leq
\tr_{\X_0}(\ket{\psi}\!\bra{\psi}).
\]
It is therefore possible to choose a density operator
$\xi\in\density{\Y\otimes\Z_0}$ such that
\[
\( 1 - \frac{1}{4k} \) \tr_{\Y} (\xi)
=
\tr_{\X_0}( \ket{\psi}\!\bra{\psi}) - 
\frac{1}{4k} \tr_{\X_0}(\ket{\phi}\!\bra{\phi}).
\]
Because $\Psi$ is admissible, we may therefore conclude that
\[
\tr_{\X_0}(\ket{\psi}\!\bra{\psi}) =
\tr_{\Y} \(  
\frac{1}{4k} \( \Psi\otimes\I_{\lin{\Z_0}} \) 
\( \ket{\phi}\!\bra{\phi} \)
+ \( 1 - \frac{1}{4k} \) \xi \),
\]
and so there must exist an admissible super-operator
$\Xi:\lin{\X_0}\rightarrow\lin{\Y}$ so that
\[
\( \Xi \otimes \I_{\lin{\Z_0}} \) (\ket{\psi}\!\bra{\psi})
=
\frac{1}{4k} \( \Psi\otimes\I_{\lin{\Z_0}} \) \( \ket{\phi}\!\bra{\phi} \)
+ \( 1 - \frac{1}{4k} \) \xi.
\]
Consequently we have
\[
\mu(R) \geq \ip{R}{J(\Xi)}
= \ip{\Pi}{\( \Xi \otimes \I_{\lin{\Z_0}} \) (\ket{\psi}\!\bra{\psi})}
\geq \frac{1}{4k}\ip{\Pi}{\(
  \Psi\otimes\I_{\lin{\Z_0}}\)\(\ket{\phi}\!\bra{\phi}\)}
= \frac{1}{4k} \mu(S)
\]
as required.

Having obtained a uniform vector $\ket{\phi}$ that gives (when
combined with $\Pi$) an interactive measurement operator $S$ satisfying
\eqref{eq:bounds-on-S}, we must make a couple of additional
modifications to be sure that a well-conditioned interactive
measurement operator has been obtained.

First, we will replace $\X_0$ and $\Z_0$ with the spaces
$\X = \complex^N$ and $\Z = \complex^N$ for 
$N = \abs{\Sigma}$.
Let $X\in\unitary{\X,\X_0}$ and $Z\in\unitary{\Z,\Z_0}$ be linear
isometries defined as
\[
X = \sum_{j = 1}^N \ket{x_j}\!\bra{j}
\quad\quad\text{and}\quad\quad
Z = \sum_{j = 1}^N \ket{z_j}\!\bra{j},
\]
and define
\[
\ket{\tau} = \( X^{\ast} \otimes Z^{\ast} \)\ket{\phi}
= \frac{1}{\sqrt{N}}\sum_{j = 1}^N\ket{j}\ket{j}
\quad\quad\text{and}\quad\quad
P = \(\I_{\Y}\otimes Z^{\ast}\) \Pi \(\I_{\Y}\otimes Z\).
\]
It is clear that $\ket{\tau}$ is a unit vector and $P$ is an ordinary
measurement operator.
(It might not be that $P$ is a projection operator, but it is positive
semidefinite and satisfies $P\leq\I_{\Y\otimes\Z}$.)
Finally, let $Q\in\pos{\Y\otimes\X}$ be the interactive measurement
operator defined by the vector $\ket{\tau}$ and the ordinary
measurement operator
\[
\(1 - \frac{1}{64 k}\) P + \frac{1}{64 k} \I_{\Y\otimes\Z}.
\]
It holds that
\[
\mu(Q) = \( 1 - \frac{1}{64k} \) \mu(S) + \frac{1}{64 k},
\]
and therefore
\begin{equation} \label{eq:Q-and-R-bounds}
\mu(R) - \( 1 - \frac{1}{8k}\) \leq \mu(Q) \leq 4k \mu(R) +
\frac{1}{64k}.
\end{equation}

Now let us verify that $Q$ has the properties we require of it.
First let us consider the condition number $\kappa(Q)$.
It is easily shown that
\[
Q \geq \frac{1}{64kN} \I_{\Y\otimes\X}
\quad\quad\text{and}\quad\quad
Q \leq \frac{1}{N} \I_{\Y\otimes\X},
\]
and therefore $\kappa(Q) \leq 64k$.
It remains to define nonnegligible values $\gamma$ and
$\varepsilon$, and to consider their relationship to $\mu(Q)$ in the
two cases: $x\in A_{\yes}$ and $x\in A_{\no}$.

We have assumed that the original interactive proof system has
exponentially small completeness and soundness errors, and therefore
we may assume that for sufficiently large $n$ we have
\[
x \in A_{\no} \Rightarrow \mu(R) \leq \frac{1}{256\,k^2}
\quad\quad\text{and}\quad\quad
x \in A_{\yes} \Rightarrow \mu(R) \geq  1 - \frac{1}{32 k}.
\]
(Alternately we may assume these bounds hold for all $n$ by
hard-coding small inputs $x$ into the verifier.)
Thus, by the bounds \eqref{eq:Q-and-R-bounds} above, we have
\[
x \in A_{\no} \Rightarrow \mu(Q) \leq \frac{1}{32 k}
\quad\quad\text{and}\quad\quad
x \in A_{\yes} \Rightarrow \mu(Q) \geq \frac{1}{16 k}.
\]
By taking
\[
\gamma = \frac{3}{64k} = \frac{3}{64(p(n) + 1)},
\quad \varepsilon = \frac{1}{12},\quad\text{and}\quad
q(n) \geq 64k = 64(p(n) + 1),
\]
we therefore have the properties required.

\section{Verifying maximum acceptance probability}
\label{sec:Arora-Kale}

We now describe and analyze a parallel algorithm, based on the
multiplicative weights update method, to distinguish the two cases
\eqref{eq:accept-reject} from the previous section.
The algorithm operates as described in Figure~\ref{fig:algorithm}, and
the super-operators $\Phi$ and $\Phi^{\ast}$ are as defined in
Section~\ref{sec:qip2-SDP}.
\begin{figure}[t]
\noindent\hrulefill
\begin{mylist}{8mm}
\item[1.]
Let
\[
\delta = \frac{\varepsilon^2}{8\,\kappa(Q)^2}
\quad\quad\text{and}\quad\quad
T = \ceil{\frac{24\,\ln(NM)}{\varepsilon^3\gamma^3\delta}}.
\]

\item[2.]
Let $W_0 = \I_{\Y\otimes\X}$ and let $\rho_0 = W_0/\tr(W_0)$.

\item[3.]
For $t = 0, \ldots, T-1$ do:\vspace{-1mm}
\begin{mylist}{8mm}
\item[a.]
Compute a spectral decomposition of $\Phi(\rho_t)$:
\[
\Phi(\rho_t) = \sum_{j = 1}^N \lambda_j \ket{x_j}\!\bra{x_j}.
\]

\item[b.]
Let $S = \{j\in\{1,\ldots,N\}\,:\,\gamma \lambda_j > 1\}$ and
let $s = \sum_{j\in S}\lambda_j$.

\item[c.]
If $s \leq \delta \norm{Q^{-1}}$, then halt and {\it accept}.

\item[d.]
Let
\[
Y_t = \frac{1}{s} \sum_{j\in S} \ket{x_j}\!\bra{x_j}.
\]

\item[e.]
Let
\[
W_{t+1} = \exp\(-\frac{\varepsilon\gamma\delta}{2}
\,\Phi^{\ast}(Y_0+\cdots+Y_t)\),
\]
and let $\rho_{t+1} = W_{t+1}/\tr(W_{t+1})$.
\end{mylist}

\item[4.]
Halt and {\it reject}.
\end{mylist}
\noindent\hrulefill
\caption{Algorithm to test if $\mu(Q)\geq \gamma$.}
\label{fig:algorithm}
\end{figure}

\subsection{Lemmas used in the analysis}

We will need a few basic facts and three lemmas to analyze the
algorithm described in Figure~\ref{fig:algorithm}.
We begin by noting two facts concerning matrix exponentials.
First, the {\it Golden-Thompson Inequality} 
(see Section IX.3 of \cite{Bhatia97}) states that, for any two
Hermitian matrices $X$ and $Y$ of equal dimension, we have
\[
\tr\left(e^{X + Y}\right) \leq \tr \left(e^X e^Y\right).
\]
Second is the following inequality concerning the
matrix exponential of positive semidefinite matrices.

\begin{lemma} \label{lemma:exp-inequality}
Let $P$ be an operator satisfying $0\leq P\leq \I$.  
Then for every real number $\eta > 0$, it holds that
\[
\exp(-\eta P) \leq \I - \eta \exp(-\eta)P.
\]
\end{lemma}

\begin{proof}
It is sufficient to prove the inequality for $P$ replaced by a
scalar $\lambda\in[0,1]$, for then the operator inequality follows by
considering a spectral decomposition of $P$.
If $\lambda=0$ then the inequality is trivial, so assume $\lambda>0$.
By the Mean Value Theorem there exists a value
$\lambda_0\in(0,\lambda)$ such that
\[
\frac{\exp(-\eta\lambda)-1}{\lambda} = -\eta\exp(-\eta\lambda_0)
\leq -\eta \exp(-\eta),
\]
which yields the inequality.
\end{proof}

For the proofs of the remaining two lemmas, some additional notation
is used.
Suppose that $\X = \complex^N$ and $\Y = \complex^M$, and assume that
the standard bases of these spaces are 
$\{\ket{0},\ldots,\ket{N-1}\}$ and $\{\ket{0},\ldots,\ket{M-1}\}$,
respectively.
Then we define a linear mapping
\[
\op{vec}:\lin{\X,\Y} \rightarrow \Y\otimes\X
\]
by taking $\op{vec}(\ket{i}\bra{j}) = \ket{i}\ket{j}$ for each choice
of $i\in\{0,\ldots,M-1\}$ and $j\in\{0,\ldots,N-1\}$.

\begin{lemma} \label{lemma:Uhlmann-extension}
Let $\X = \complex^N$ and $\Y = \complex^M$ for positive integers $N$
and $M$.
Let $P_0,P_1\in\pos{\X}$ and let $R_0\in\pos{\X\otimes\Y}$ satisfy
$\tr_{\Y}(R_0) = P_0$.
Then there exists an operator $R_1\in\pos{\X\otimes\Y}$ such that
$\tr_{\Y}(R_1) = P_1$ and $\fid(R_0,R_1) = \fid(P_0,P_1)$.
\end{lemma}

\begin{proof}
By the monotonicity of the fidelity function, it must hold that
$\fid(R_0,R_1)\leq\fid(P_0,P_1)$ for every choice of $R_1$ satisfying
$\tr_{\Y}(R_1) = P_1$.
It therefore suffices to show that equality can be achieved.

Choose a unitary operator $V\in\unitary{\X}$ for which
$\sqrt{P_0}\sqrt{P_1}V$ is positive semidefinite.
For such a $V$ it holds that 
$\fid(P_0,P_1) = \tr(\sqrt{P_0}\sqrt{P_1} V)$.
Now let $\W = \complex^{NM}$ and let
$\ket{u_0}\in\Y\otimes\X\otimes\W$ be a purification of $R_0$.
Given that $\ket{u_0}$ also purifies $P_0$, it must take the form
\[
\ket{u_0} = \op{vec}\( \sqrt{P_0} U^{\ast}\)
\]
for some choice of a linear isometry $U\in\unitary{\X,\Y\otimes\W}$.
Finally, let
\[
R_1 = \tr_{\W}\( \op{vec}\(\sqrt{P_1}VU^{\ast}\)
\op{vec}\(\sqrt{P_1} VU^{\ast}\)^{\ast}\).
\]
It holds that
\[
\fid(R_0,R_1) 
\geq
\abs{\ip{\sqrt{P_0}U^{\ast}}{\sqrt{P_1}VU^{\ast}}}
= \tr\(\sqrt{P_0} \sqrt{P_1}V\)
= \fid(P_0,P_1)
\]
as required.
\end{proof}

\begin{remark}
  Lemma~\ref{lemma:Uhlmann-extension} is a fairly straightforward
  extension of Uhlmann's Theorem \cite{Uhlmann76}.
  (See also pages 410--411 of \cite{NielsenC00}.)
\end{remark}

\begin{lemma} \label{lemma:FvdG-extension}
Let $R_0,R_1\in\pos{\W}$ for $\W = \complex^k$.
Then
\[
\norm{R_0 - R_1}_1 \leq 
\sqrt{2 \tr(R_0)^2 + 2 \tr(R_1)^2 - 4 \fid(R_0,R_1)^2}.
\]
\end{lemma}

\begin{proof}
Choose $V\in\unitary{\W}$ so that $\sqrt{R_0}\sqrt{R_1}V$ is positive
semidefinite, and therefore 
$\fid(R_0,R_1) = \tr(\sqrt{R_0}\sqrt{R_1}V)$.
We have that
\[
\op{vec}\(\sqrt{R_0}\) \in \W\otimes\W
\quad\quad\text{and}\quad\quad
\op{vec}\(\sqrt{R_1} V\) \in \W\otimes\W
\]
purify $R_0$ and $R_1$, respectively, so by the monotonicity of the
trace norm it holds that
\[
\norm{R_0 - R_1}_1
\leq
\norm{
\op{vec}\(\sqrt{R_0}\) \op{vec}\(\sqrt{R_0}\)^{\ast}
-
\op{vec}\(\sqrt{R_1} V\) \op{vec}\(\sqrt{R_1} V\)^{\ast}}_1.
\]
Using the inequality
\[
\norm{A}_1 \leq \sqrt{\op{rank}(A)} \norm{A}_2,
\]
which follows easily from the expressions of the trace and Frobenius
norms in terms of singular values, along with the Cauchy--Schwarz
inequality, we have
\begin{align*}
\norm{R_0 - R_1}_1 & \leq \sqrt{2}
\norm{
\op{vec}\(\sqrt{R_0}\) \op{vec}\(\sqrt{R_0}\)^{\ast}
-
\op{vec}\(\sqrt{R_1} V\) \op{vec}\(\sqrt{R_1} V\)^{\ast}}_2\\
& =
\sqrt{2 \tr(R_0)^2 + 2 \tr(R_1)^2 - 4
  \tr\(\sqrt{R_0}\sqrt{R_1}V\)^2}\\
& = \sqrt{2 \tr(R_0)^2 + 2 \tr(R_1)^2 - 4 \fid(R_0,R_1)^2}
\end{align*}
as required.
\end{proof}

\begin{remark}
When $R_0 = \rho_0$ and $R_1 = \rho_1$ for density operators $\rho_0$
and $\rho_1$, we obtain the familiar inequality
\[
\norm{\rho_0 - \rho_1}_1 \leq 2 \sqrt{1 - \fid(\rho_0,\rho_1)^2},
\]
or equivalently
\[
\fid(\rho_0,\rho_1) \leq \sqrt{1 - \frac{1}{4}\norm{\rho_0 -
    \rho_1}_1^2},
\]
which is one of the Fuchs-van de Graaf inequalities \cite{FuchsvdG99}.
\end{remark}

\subsection{Analysis of the algorithm (ignoring precision)}

Our algorithm cannot be implemented exactly using bounded-depth
Boolean circuits: the spectral decompositions and matrix
exponentials can only be approximated.
However, for the sake of exposition, the issue of precision will be
completely ignored in this subsection; meaning that we will
imagine that all of the operations can be performed exactly. 
In the section that follows this one, the actual precision
requirements of the algorithm are considered.
As is shown there, it turns out that the algorithm is not particularly
sensitive to errors, and in fact it is possible to perform all of the
required computations in parallel with exponentially greater precision
than would be required for the correctness of the algorithm.

Let us consider first the case that the algorithm accepts.
Let $\rho = \rho_t$ for the iteration $t$ of the loop in step 3 in
which acceptance occurs.
To prove that the algorithm has answered correctly, we will construct
an operator $X\in\pos{\Y\otimes\X}$ such that $\Phi(X) \leq \I_{\Y}$
and $\tr(X) \geq (1 - \varepsilon)\gamma$; and therefore
$\mu(Q) \geq (1 - \varepsilon)\gamma$.
By the conditions \eqref{eq:accept-reject} on $Q$, this implies that
$\mu(Q) \geq (1 + 4\varepsilon)\gamma$, and therefore $x\in A_{\yes}$,
as required.

The operator $X$ is defined as follows.
First, let $R_0 = Q^{-1/2}\rho Q^{-1/2}$, let
$P_0 = \tr_{\Y}(R_0) = \Phi(\rho)$, and let
\[
P_1 = \frac{1}{\gamma}\sum_{j\in S} \ket{x_j}\!\bra{x_j} + 
\sum_{j\not\in S}\lambda_j \ket{x_j}\!\bra{x_j}.
\]
By Lemma~\ref{lemma:Uhlmann-extension} there must exist
$R_1\in\pos{\Y\otimes\X}$ such that $\tr_{\Y}(R_1) = P_1$ and
$\fid(R_0,R_1) = \fid(P_0,P_1)$.
We then take
\[
X = \gamma \sqrt{Q} R_1 \sqrt{Q}.
\]
It holds that $X \geq 0$ and $\Phi(X) \leq \I_{\X}$.
To establish a lower bound on $\tr(X)$, we first note that
\[
1 - \tr\(\sqrt{Q} R_1 \sqrt{Q}\)
= \ip{Q}{R_0 - R_1}
\leq \norm{Q} \norm{R_0 - R_1}_1.
\]
By Lemma~\ref{lemma:FvdG-extension}, we conclude that
\[
\norm{R_0 - R_1}_1 
\leq \sqrt{ 2 \tr(P_0)^2 + 2 \tr(P_1)^2 - 4 \fid(P_0,P_1)^2}.
\]
It holds that $\tr(P_1) \leq \tr(P_0)$, and we also have
\[
\fid(P_0,P_1) 
\geq \sum_{j\not\in S}\lambda_j
= \tr(P_0) - s \geq \tr(P_0) - \delta \snorm{Q^{-1}}.
\]
Therefore, given that $\tr(P_0) \leq \snorm{Q^{-1}}$, we have
$\norm{R_0 - R_1}_1 \leq \sqrt{8 \delta} \snorm{Q^{-1}}$,
and so
\[
1 - \tr\(\sqrt{Q} R_1 \sqrt{Q}\) \leq 
\sqrt{8 \delta} \snorm{Q^{-1}} \norm{Q}
= \sqrt{8 \delta} \kappa(Q) = \varepsilon.
\]
It follows that
\[
\tr(X) = \gamma \tr\(\sqrt{Q} R_1 \sqrt{Q}\)
\geq (1 - \varepsilon)\gamma
\]
as required.

Now let us consider the case that the algorithm rejects.
Along similar lines to the previous case, we will construct an
operator $Y\in\pos{\X}$ such that 
$\Phi^{\ast}(Y) \geq \I_{\Y\otimes\X}$ and
$\tr(Y) \leq (1 + \varepsilon)\gamma$.
By the conditions \eqref{eq:accept-reject} on $Q$ this implies that
$\mu(Q) \leq (1 - 4\varepsilon)\gamma$, and therefore $x\in A_{\no}$.
In particular, we may take
\[
Y = \frac{1+\varepsilon}{T}(Y_0+\cdots+Y_{T-1}).
\]
Each operator $Y_t$ satisfies
\[
\tr(Y_t) = \frac{\abs{S}}{s} < \frac{1}{s}\sum_{j\in S} \gamma \lambda_j
= \gamma,
\]
and therefore $\tr(Y) < (1 + \varepsilon)\gamma$.
Each $Y_t$ is also clearly positive semidefinite, so it remains to
prove that $\Phi^{\ast}(Y) \geq \I_{\Y\otimes\X}$.
To this end we will first establish two conditions on each operator
$Y_t$.
First, we have
\begin{equation} \label{eq:Y-condition-1}
\ip{\rho_t}{\Phi^{\ast}(Y_t)} = \ip{\Phi(\rho_t)}{Y_t} 
= \frac{1}{s}\sum_{j\in S}\lambda_j = 1.
\end{equation}
Second, given that $s > \delta\snorm{Q^{-1}}$ for the case at hand,
we have
\[
\norm{\Phi^{\ast}(Y_t)} 
= \snorm{Q^{-1/2}(\I_{\Y}\otimes Y_t) Q^{-1/2}}
\leq \snorm{Q^{-1}} \norm{Y_t} 
= \frac{\snorm{Q^{-1}}}{s} 
< \frac{1}{\delta},
\]
and therefore $\norm{\delta \Phi^{\ast}(Y_t)} < 1$.

Now, for the sake of clarity, let us write $\eta = \varepsilon\gamma/2$.
Note that, for $0\leq t\leq T-1$, it holds that
\begin{align*}
\tr(W_{t+1})
& = \tr\left[\exp( - \eta \delta \Phi^{\ast}(Y_0+\cdots+Y_t))\right]\\
& \leq 
\tr \left[\exp( - \eta \delta \Phi^{\ast}(Y_0+\cdots+Y_{t-1}))
\exp(-\eta\delta\Phi^{\ast}(Y_t))\right]\\
& = \tr\left[W_t\exp(-\eta\delta\Phi^{\ast}(Y_t))\right],
\end{align*}
where we have used the Golden--Thompson Inequality.
Given that $\norm{\delta\Phi^{\ast}(Y_t)}\leq 1$ we have by
Lemma~\ref{lemma:exp-inequality} that
\[
\exp(-\eta\delta \Phi^{\ast}(Y_t)) \leq
\I - \eta\delta \exp(-\eta) \Phi^{\ast}(Y_t),
\]
and therefore
\[
\tr(W_{t+1}) \leq \tr(W_t) (1 - \eta\delta \exp(-\eta) 
\ip{\rho_t}{\Phi^{\ast}(Y_t)})
\leq \tr(W_t) \exp(-\eta\delta\exp(-\eta)).
\]
(Here we have used the inequality $\exp(-\alpha) \geq 1 - \alpha$,
which holds for all real numbers $\alpha$, as well as the fact that
$\tr(AB) \leq \tr(AC)$ whenever $A\geq 0$ and $B\leq C$).
Repeating this argument, and substituting $\tr(W_0) = NM$, we have
\[
\tr(W_T) \leq NM \exp(-\eta\delta T \exp(-\eta)).
\]
On the other hand, it is clear that
\[
\tr(W_T) 
= \tr\left[\exp(-\eta\delta \Phi^{\ast}(Y_0+\cdots+Y_{T-1}))\right]
\geq \exp(-\eta\delta\lambda_{NM}(\Phi^{\ast}(Y_0+\cdots+Y_{T-1}))),
\]
and therefore
\[
\lambda_{NM}\( \Phi^{\ast}\(\frac{Y_0 + \cdots + Y_{T-1}}{T}\) \)
\geq \exp(-\eta) - \frac{\ln(NM)}{\eta\delta T}.
\]
Substituting the specified value of $T$, and using the fact that
$\exp(-\eta) - \frac{\eta^2}{3} \geq 1 - \eta$ (which holds for 
any $\eta\in[0,1]$), we have
\[
\lambda_{NM}\( \Phi^{\ast}\(\frac{Y_0 + \cdots + Y_{T-1}}{T}\) \)
\geq 1 - \eta = 1 - \frac{\varepsilon\gamma}{2}.
\]
Therefore
\[
\lambda_{NM}(\Phi^{\ast}(Y)) \geq (1 + \varepsilon)\(1 -
\frac{\varepsilon\gamma}{2}\)
\geq 1,
\]
and so $\Phi^{\ast}(Y) \geq \I_{\Y\otimes\X}$ as required.

We therefore have that the algorithm works correctly, modulo the
precision issues to be discussed in the section following this one.
It remains only to observe that it can be implemented in
$\class{NC}$ (meaning that it results in an
$\class{NC}(\mathit{poly})$ computation when composed with the first
two steps of the simulation).
Some of the details required to argue this can be found below;
but at a high level one sees that each iteration of the loop in step 3
can be performed with high precision in $\class{NC}$, and the total
number of iterations required is polynomial in $n$ (and therefore
polylogarithmic in the size of $Q$).

\section{Precision requirements for the simulation} \label{sec:precision}

We now discuss the precision that is required for the simulation
to yield a correct answer.
It will turn out that the simulation is not particularly sensitive to
errors, and one could in fact afford to take exponentially more
precision than is required and still be within the class
$\class{NC}(\mathit{poly})$.

The first step of the simulation, in which an explicit description of
$\ket{\psi}$ and $\Pi$ is obtained, can be performed exactly
in $\class{NC}(\mathit{poly})$ as has already been observed.
So, let us move on to the second step, in which $\ket{\psi}$ and $\Pi$
are processed to obtain a well-conditioned interactive measurement
operator~$Q$.
This step requires the approximation of one singular value
decomposition (to approximate the Schmidt decomposition of
$\ket{\psi}$), along with a few other operations that can be performed
exactly or with high precision in $\class{NC}$.

For the moment let us denote by $\widetilde{Q}$ the actual operator
that is computed by an $\class{NC}$ implementation of this step, as
opposed to the true operator $Q$ that would be output by an idealized,
exact complex number algorithm.
By computing the singular value decomposition to high precision, it is
possible to take such an approximation so that
\[
\norm{Q - \widetilde{Q}} < 2^{-2^{\mathit{poly}(n)}},
\]
where $\mathit{poly}$ denotes any polynomial of our choice.
It is not difficult to prove that the quantity
$|\mu(Q) - \mu(\widetilde{Q})|$ is upper-bounded by $N$ times
$\snorm{Q - \widetilde{Q}}$, and therefore we may take $\widetilde{Q}$
so that
\[
\abs{\mu(Q) - \mu(\widetilde{Q})} < 2^{-2^{\mathit{poly}(n)}},
\]
again for $\mathit{poly}$ denoting any polynomial of our choice.
We do not need this much precision: we only need a
$1/\mathit{poly}(n)$ separation between the values of
$\mu(\widetilde{Q})$ for the cases $x\in A_{\yes}$ and $x\in A_{\no}$,
which requires that $\snorm{Q - \widetilde{Q}}$ is exponentially
(rather than double-exponentially) small in $n$.
So, to be concrete, we may decide to take sufficient precision so that
\[
\norm{Q - \widetilde{Q}} < \varepsilon\gamma/N,
\]
and therefore
\[
\abs{\mu(Q) - \mu(\widetilde{Q})} < \varepsilon \gamma.
\]
Thus,
\begin{equation} \label{eq:approximate-Q-bounds}
\begin{split}
x\in A_{\yes} \quad & \Rightarrow\quad \mu(\widetilde{Q}) 
\geq (1 + 3\varepsilon)\gamma\\[2mm]
x\in A_{\no} \quad & \Rightarrow\quad \mu(\widetilde{Q}) 
\leq (1 - 3\varepsilon)\gamma.
\end{split}
\end{equation}

Hereafter we will return to writing $Q$ rather than $\widetilde{Q}$,
with the understanding that $Q$ now represents an approximation that
is stored by our algorithm.
In addition, we will assume that $\sqrt{Q}$, and therefore $Q^{-1/2}$
as well, has Gaussian rational entries and is known precisely.
This assumption is easily met by replacing $Q$ with the square of a
high precision approximation to $\sqrt{Q}$.
The point of this assumption is that we avoid having to consider an
additional error term every time $\sqrt{Q}$ or $Q^{-1/2}$ is involved
in any computation.
(There is no reason beyond simplifying the analysis to make this
assumption.)

Now suppose that the algorithm described in
Section~\ref{sec:Arora-Kale} is performed with limited precision.
Consider first the case that the algorithm accepts, and let
$\rho = \rho_t$ denote the density operator that is stored by the
algorithm on the iteration $t$ in which acceptance occurs.
Note that it is not necessary to view that $\rho$ is an approximation
of something else: the simple fact that $\rho$ causes acceptance will
allow us to conclude that $x\in A_{\yes}$ in a similar way to the
error-free analysis.

Specifically, we will construct an operator
$X\in\pos{\Y\otimes\X}$ such that $\Phi(X) \leq \I_{\Y}$ and 
$\tr(X) \geq (1 - 2 \varepsilon)\gamma$.
As before we will let $R_0 = Q^{-1/2}\rho Q^{-1/2}$ and
$P_0 = \tr_{\Y}(R_0)$.
This time, we must consider that the spectral decomposition is
approximate.
Let us write
\[
\widetilde{P}_0 = \sum_{j = 1}^N \lambda_j \ket{x_j}\!\bra{x_j}
\]
to denote the approximate value of the spectral decomposition, so that
$\snorm{P_0 - \widetilde{P}_0}$ represents the error in this
approximation, and let us assume that sufficient accuracy is taken so
that
\begin{equation} \label{eq:accept-approx}
\norm{P_0 - \widetilde{P}_0} < \frac{\delta}{4 N \norm{Q^{-1}}}.
\end{equation}
As above, this is significantly less accuracy than is available---for
we could take
\[
\norm{P_0 - \widetilde{P}_0} < 2^{-2^{\mathit{poly}(n)}}
\]
if it were advantageous.
Continuing on as before, let
\[
P_1 = \frac{1}{\gamma}\sum_{j\in S} \ket{x_j}\!\bra{x_j} + 
\sum_{j\not\in S}\lambda_j \ket{x_j}\!\bra{x_j},
\]
let $R_1\in\pos{\Y\otimes\X}$ to be an extension of $P_1$
for which $\fid(R_0,R_1) = \fid(P_0,P_1)$, and let
\[
X = \gamma \sqrt{Q} R_1 \sqrt{Q}.
\]

We have that $X\geq 0$ and $\Phi(X)\leq \I_{\X}$ as before, and to
establish a lower bound on $\tr(X)$ we again use the fact that 
\[
1 - \tr\(\sqrt{Q} R_1 \sqrt{Q}\) \leq \norm{Q} \norm{R_0 - R_1}_1,
\]
as well as the bound
\[
\norm{R_0 - R_1}_1 
\leq \sqrt{2 \tr(P_0)^2 + 2 \tr(P_1)^2 - 4 \fid(P_0,P_1)^2}.
\]
Based on the bound \eqref{eq:accept-approx}, it follows that
\[
\tr(P_1)^2 \leq \tr(P_0)^2 + \delta
\]
and
\[
\fid(P_0,P_1)^2 \geq \(\tr(P_0) - \delta \norm{Q^{-1}}\)^2 + \delta,
\]
and therefore
\[
\norm{R_0 - R_1}_1 \leq \sqrt{14 \delta} \snorm{Q^{-1}}.
\]
It follows that
\[
1 - \tr\(\sqrt{Q} R_1 \sqrt{Q}\) < 2\varepsilon,
\]
and therefore
\[
\tr(X) = \gamma \tr\(\sqrt{Q} R_1 \sqrt{Q}\)
\geq (1 - 2\varepsilon)\gamma
\]
as required.

Now let us consider the case that the algorithm rejects, and 
in particular let us focus on the operators $Y_0,\ldots,Y_{T-1}$ that
are computed over the course of the algorithm.
As for the case of acceptance, these operators are not viewed as
approximations to anything: the fact that these operators exist and
cause rejection in the algorithm is enough to conclude that
$x\in A_{\no}$.
Let us continue to write
\[
W_{t+1} = \exp\(-\frac{\varepsilon\gamma\delta}{2} \Phi^{\ast}(Y_0 +
\cdots + Y_t)\)
\]
and $\rho_t = W_t/\tr(W_t)$ for each $t = 0,\ldots,T-1$;
but we must keep in mind that the algorithm only computes
approximations of these operators.
The algorithm must also approximate the spectral decomposition of
each $\Phi(\rho_t)$, where the source of errors in this case comes
from both the spectral decomposition computation and the fact
that $\rho_t$ is approximated.

Now, in the error free analysis, the conditions
\[
\tr(Y_t) \leq \gamma,\quad\quad
\ip{\rho_t}{\Phi^{\ast}(Y_t)} = 1,
\quad\quad\text{and}\quad\quad
\norm{\delta \Phi^{\ast}(Y_t)} \leq 1
\]
were proved, and these conditions allowed us to conclude that $Y$
satisfies
\[
\Phi^{\ast}(Y) \geq \I_{\Y\otimes\X}
\quad\quad\text{and}\quad\quad
\tr(Y) \leq (1 + \varepsilon)\gamma.
\]
If we follow precisely the same proof, but with the condition
$\ip{\rho_t}{\Phi^{\ast}(Y_t)} = 1$ replaced by
\[
\ip{\rho_t}{\Phi^{\ast}(Y_t)} \geq 1 - \alpha
\]
for some choice of $\alpha>0$, we once again find immediately that 
$\tr(Y) \leq (1 + \varepsilon)\gamma$.
This time we have
\[
\lambda_{NM}\( \Phi^{\ast}\(\frac{Y_0 + \cdots + Y_{T-1}}{T}\) \)
\geq (1 - \alpha)\exp(-\eta) - \frac{\ln(NM)}{\eta\delta T},
\]
but under the assumption $\alpha < \eta^2/12$, say, it follows again
that $\Phi^{\ast}(Y) \geq \I_{\Y\otimes\X}$.

Thus, given that the conditions $\tr(Y_t) \leq \gamma$ and
$\norm{\delta\Phi^{\ast}(Y_t)}\leq 1$ follow from an inspection of the
algorithm as before, it suffices to compute the matrix exponentials
and spectral decompositions with sufficient accuracy that
$\ip{\rho_t}{\Phi^{\ast}(Y_t)} > 1 - \eta^2/12$.
This is easily done: as the argument of the matrix exponentials have
norm bounded by $T$, one is able to compute both the matrix
exponentials and the spectral decompositions with exponentially
greater accuracy in $\class{NC}$ than is required.

\section{Conclusion}
\label{sec:conclusion}

We have proved that $\class{QIP}(2)\subseteq\class{PSPACE}$ using a
semidefinite programming formulation of the maximum acceptance
probability of two-message quantum interactive proof systems, along
with the multiplicative weights update method for verifying these
values.
An obvious question remains: can this method be extended, or some
other method devised, to prove $\class{QIP} = \class{PSPACE}$?

\subsection*{Acknowledgments}

The research presented in this paper was supported by Canada's NSERC,
MITACS, the Canadian Institute for Advanced Research, and the US Army
Research Office and National Security Agency; and was partially done
while Rahul Jain was a member of the Institute for Quantum Computing
and the School of Computer Science at the University of Waterloo.

\bibliographystyle{alpha}

\begin{thebibliography}{BOGKW88}

\bibitem[AHK05]{AroraHK05b}
S.~Arora, E.~Hazan, and S.~Kale.
\newblock Fast algorithms for approximate semidefinite programming using the
  multiplicative weights update method.
\newblock In {\em Proceedings of the 46th Annual IEEE Symposium on Foundations
  of Computer Science}, pages 339--348, 2005.

\bibitem[AK07]{AroraK07}
S.~Arora and S.~Kale.
\newblock A combinatorial, primal-dual approach to semidefinite programs.
\newblock In {\em Proceedings of the Thirty-Ninth Annual ACM Symposium on
  Theory of Computing}, pages 227--236, 2007.

\bibitem[Bab85]{Babai85}
L.~Babai.
\newblock Trading group theory for randomness.
\newblock In {\em Proceedings of the 17th Annual ACM Symposium on Theory of
  Computing}, pages 421--429, 1985.

\bibitem[BATS09]{Ben-AroyaT09}
A.~Ben-Aroya and A.~Ta-Shma.
\newblock On the complexity of approximating the diamond norm.
\newblock Available as arXiv.org e-Print 0902.3397, 2009.

\bibitem[BCP83]{BorodinCP83}
A.~Borodin, S.~Cook, and N.~Pippenger.
\newblock Parallel computation for well-endowed rings and space-bounded
  probabilistic machines.
\newblock {\em Information and Control}, 58:113--136, 1983.

\bibitem[BGH82]{BorodinGH82}
A.~Borodin, J.~von~zur Gathen, and J.~Hopcroft.
\newblock Fast parallel matrix and {GCD} computations.
\newblock In {\em Proceedings of the 23rd Annual IEEE Symposium on Foundations
  of Computer Science}, pages 65--71, 1982.

\bibitem[Bha97]{Bhatia97}
R.~Bhatia.
\newblock {\em Matrix Analysis}.
\newblock Springer, 1997.

\bibitem[BM88]{BabaiM88}
L.~Babai and S.~Moran.
\newblock {A}rthur-{M}erlin games: a randomized proof system, and a hierarchy
  of complexity classes.
\newblock {\em Journal of Computer and System Sciences}, 36(2):254--276, 1988.

\bibitem[BOFKT86]{BenOrFKT86}
M.~Ben-Or, E.~Feig, D.~Kozen, and P.~Tiwari.
\newblock A fast parallel algorithm for determining all roots of a polynomial
  with real roots.
\newblock In {\em Proceedings of the 18th Annual ACM Symposium on Theory of
  Computing}, pages 340--349, 1986.

\bibitem[BOGKW88]{Ben-OrGKW88}
M.~Ben-Or, S.~Goldwasser, J.~Kilian, and A.~Wigderson.
\newblock Multi-prover interactive proofs: how to remove intractability
  assumptions.
\newblock In {\em Proceedings of the 20th Annual ACM Symposium on Theory of
  Computing}, pages 113--131, 1988.

\bibitem[Bor77]{Borodin77}
A.~Borodin.
\newblock On relating time and space to size and depth.
\newblock {\em SIAM Journal on Computing}, 6:733--744, 1977.

\bibitem[Cho75]{Choi75}
M.-D. Choi.
\newblock Completely positive linear maps on complex matrices.
\newblock {\em Linear Algebra and Its Applications}, 10(3):285--290, 1975.

\bibitem[Csa76]{Csanky76}
L.~Csanky.
\newblock Fast parallel matrix inversion algorithms.
\newblock {\em SIAM Journal on Computing}, 5(4):618--623, 1976.

\bibitem[FvdG99]{FuchsvdG99}
C.~Fuchs and J.~van~de Graaf.
\newblock Cryptographic distinguishability measures for quantum-mechanical
  states.
\newblock {\em IEEE Transactions on Information Theory}, 45(4):1216--1227,
  1999.

\bibitem[Gat93]{vzGathen93}
J.~von~zur Gathen.
\newblock Parallel linear algebra.
\newblock In J.~Reif, editor, {\em Synthesis of Parallel Algorithms},
  chapter~13. Morgan Kaufmann Publishers, Inc., 1993.

\bibitem[GLS93]{GrotschelLS93}
M.~Gr\"otschel, L.~Lov\'asz, and A.~Schrijver.
\newblock {\em Geometric Algorithms and Combinatorial Optimization}.
\newblock Springer--Verlag, second corrected edition, 1993.

\bibitem[GMR85]{GoldwasserMR85}
S.~Goldwasser, S.~Micali, and C.~Rackoff.
\newblock The knowledge complexity of interactive proof systems.
\newblock In {\em Proceedings of the 17th Annual ACM Symposium on Theory of
  Computing}, pages 291--304, 1985.

\bibitem[GMR89]{GoldwasserMR89}
S.~Goldwasser, S.~Micali, and C.~Rackoff.
\newblock The knowledge complexity of interactive proof systems.
\newblock {\em SIAM Journal on Computing}, 18(1):186--208, 1989.

\bibitem[GMW91]{GoldreichMW91}
O.~Goldreich, S.~Micali, and A.~Wigderson.
\newblock Proofs that yield nothing but their validity or all languages in {NP}
  have zero-knowledge proof systems.
\newblock {\em Journal of the ACM}, 38(1):691--729, 1991.

\bibitem[GS89]{GoldwasserS89}
S.~Goldwasser and M.~Sipser.
\newblock Private coins versus public coins in interactive proof systems.
\newblock In S.~Micali, editor, {\em Randomness and Computation}, volume~5 of
  {\em Advances in Computing Research}, pages 73--90. JAI Press, 1989.

\bibitem[GW07]{GutoskiW07}
G.~Gutoski and J.~Watrous.
\newblock Toward a general theory of quantum games.
\newblock In {\em Proceedings of the 39th Annual ACM Symposium on Theory of
  Computing}, pages 565--574, 2007.

\bibitem[HKSZ08]{HallgrenKSZ08}
S.~Hallgren, A.~Kolla, P.~Sen, and S.~Zhang.
\newblock Making classical honest verifier zero knowledge protocols secure
  against quantum attacks.
\newblock In {\em Proceedings of the 35th International Colloquium on Automata,
  Languages and Programming}, volume 5126 of {\em Lecture Notes in Computer
  Science}, pages 592--603. Springer, 2008.

\bibitem[Jam72]{Jamiolkowski72}
A.~Jamio{\l}kowski.
\newblock Linear transformations which preserve trace and positive
  semidefiniteness of operators.
\newblock {\em Reports on Mathematical Physics}, 3(4):275--278, 1972.

\bibitem[JW09]{JainW09}
R.~Jain and J.~Watrous.
\newblock Parallel approximation of non-interactive zero-sum quantum games.
\newblock In {\em Proceedings of the 24th IEEE Conference on Computational
  Complexity}, 2009.
\newblock To appear.

\bibitem[Kal07]{Kale07}
S.~Kale.
\newblock {\em Efficient algorithms using the multiplicative weights update
  method}.
\newblock PhD thesis, Princeton University, 2007.

\bibitem[KKMV08]{KempeKMV08}
J.~Kempe, H.~Kobayashi, K.~Matsumoto, and T.~Vidick.
\newblock Using entanglement in quantum multi-prover interactive proofs.
\newblock In {\em Proceedings of the 23rd Annual Conference on Computational
  Complexity}, 2008.

\bibitem[KM03]{KobayashiM03}
H.~Kobayashi and K.~Matsumoto.
\newblock Quantum multi-prover interactive proof systems with limited prior
  entanglement.
\newblock {\em Journal of Computer and System Sciences}, 66(3), 2003.

\bibitem[Kob08]{Kobayashi08}
H.~Kobayashi.
\newblock General properties of quantum zero-knowledge proofs.
\newblock In {\em Proceedings of the Fifth IACR Theory of Cryptography
  Conference}, volume 4948 of {\em Lecture Notes in Computer Science}, pages
  107--124. Springer, 2008.

\bibitem[KW00]{KitaevW00}
A.~Kitaev and J.~Watrous.
\newblock Parallelization, amplification, and exponential time simulation of
  quantum interactive proof system.
\newblock In {\em Proceedings of the 32nd Annual ACM Symposium on Theory of
  Computing}, pages 608--617, 2000.

\bibitem[LFKN92]{LundFKN92}
C.~Lund, L.~Fortnow, H.~Karloff, and N.~Nisan.
\newblock Algebraic methods for interactive proof systems.
\newblock {\em Journal of the ACM}, 39(4):859--868, 1992.

\bibitem[MW05]{MarriottW05}
C.~Marriott and J.~Watrous.
\newblock Quantum {Arthur-Merlin} games.
\newblock {\em Computational Complexity}, 14(2):122--152, 2005.

\bibitem[NC00]{NielsenC00}
M.~A. Nielsen and I.~L. Chuang.
\newblock {\em Quantum Computation and Quantum Information}.
\newblock Cambridge University Press, 2000.

\bibitem[Nef94]{Neff94}
C.~A. Neff.
\newblock Specified precision polynomial root isolation is in {NC}.
\newblock {\em Journal of Computer and System Sciences}, 48(3):429--463, 1994.

\bibitem[Sha92]{Shamir92}
A.~Shamir.
\newblock {IP} $=$ {PSPACE}.
\newblock {\em Journal of the ACM}, 39(4):869--877, 1992.

\bibitem[Uhl76]{Uhlmann76}
A.~Uhlmann.
\newblock The "transition probability" in the state space of a $\ast$-algebra.
\newblock {\em Reports on Mathematical Physics}, 9(2):273--279, 1976.

\bibitem[Wat02]{Watrous02}
J.~Watrous.
\newblock Limits on the power of quantum statistical zero-knowledge.
\newblock In {\em Proceedings of the 43rd Annual IEEE Symposium on Foundations
  of Computer Science}, pages 459--468, 2002.

\bibitem[Wat03]{Watrous03-pspace}
J.~Watrous.
\newblock {PSPACE} has constant-round quantum interactive proof systems.
\newblock {\em Theoretical Computer Science}, 292(3):575--588, 2003.

\bibitem[Wat06]{Watrous06}
J.~Watrous.
\newblock Zero-knowledge against quantum attacks.
\newblock In {\em Proceedings of the 38th Annual ACM Symposium on Theory of
  Computing}, pages 296--305, 2006.

\bibitem[Wat09a]{Watrous09-complexity}
J.~Watrous.
\newblock Quantum computational complexity.
\newblock In {\em Encyclopedia of Complexity and System Science}. Springer,
  2009.
\newblock Available as arXiv.org e-Print 0804.3401.

\bibitem[Wat09b]{Watrous09-cb-sdp}
J.~Watrous.
\newblock Semidefinite programs for completely bounded norms.
\newblock Available as arXiv.org e-Print 0901.4709, 2009.

\bibitem[Weh06]{Wehner06}
S.~Wehner.
\newblock Entanglement in interactive proof systems with binary answers.
\newblock In {\em Proceedings of the 23rd Annual Symposium on Theoretical
  Aspects of Computer Science}, volume 3884 of {\em Lecture Notes in Computer
  Science}, pages 162--171. Springer, 2006.

\bibitem[WK06]{WarmuthK06}
M.~Warmuth and D.~Kuzmin.
\newblock Online variance minimization.
\newblock In {\em Proceedings of the 19th Annual Conference on Learning
  Theory}, volume 4005 of {\em Lecture Notes in Computer Science}, pages
  514--528. Springer, 2006.

\end{thebibliography}


\end{document}